\immediate\write18{makeindex \jobname.nlo -s nomencl.ist -o \jobname.nls}
\documentclass[10pt,a4paper,twoside]{article}
	\usepackage[backend=bibtex,style=numeric,natbib=true,ibidtracker=strict]{biblatex} 
	\usepackage{nomencl} 
	\makenomenclature
	\usepackage{marginnote}
	\usepackage[margin=3cm,marginparwidth=2.5cm]{geometry}
	\usepackage{hyperref}
	\usepackage{color}
	\usepackage{enumerate}
	\usepackage{tikz}
	\usepackage{amsmath}
	\usepackage{amssymb}
	\usepackage{amsthm}
	\usepackage{dsfont}
	\usepackage{slashed}
	\newcommand{\ncom}{\newcommand}
	\newcommand{\rncom}{\renewcommand}
	\ncom{\af}{\mathbf{a}}
	\ncom{\bv}{\mathbf{b}}
	\ncom{\syml}{\nomenclature}
	\ncom{\innote}[1]{\reversemarginpar\marginnote{\textcolor{red}{#1}}}
	\ncom{\outnote}[1]{\normalmarginpar\marginnote{\textcolor{red}{#1}}}
	\ncom{\h}{\mathcal{H}}
	\ncom{\A}{\mathcal{A}}
	\ncom{\dee}{\mathrm{d}}
	\ncom{\dbar}{{\declareslashed{}{\text{--}}{0.04}{0.2}{\mathrm{d}}\slashed{\mathrm{d}}}}
	\ncom{\desda}{\Leftrightarrow}
	\ncom{\ket}[1]{\left| #1 \right\rangle}
	\ncom{\bra}[1]{\left\langle #1 \right|}
	\ncom{\braket}[2]{\left< #1\middle| #2\right>}
	\ncom{\set}[2]{\left\{ #1\:\middle|\: #2\right\}}
	\DeclareCiteCommand{\citeyear}{}{\bibhyperref{\printdate}}{, }{} 

	\DeclareMathAlphabet{\mathpzc}{OT1}{pzc}{m}{it}
	\rncom{\nomname}{List of symbols}

	\theoremstyle{plain}
	\newtheorem{thm}{Theorem}
	\newtheorem{lemma}{Lemma}
	
	\theoremstyle{definition}

\begin{document}

\title{Conway-Kochen and the Finite Precision Loophole}
\author{Ronnie Hermens\thanks{Department of Theoretical Philosophy, University of Groningen, Oude Bo\-te\-ringe\-straat 52, 9712 GL Gro\-ning\-en, Netherlands}}

\maketitle

\begin{abstract}
Recently Cator \& Landsman made a comparison between Bell's Theorem and Conway \& Kochen's Strong Free Will Theorem.
Their overall conclusion was that the latter is stronger in that it uses fewer assumptions, but also that it has two shortcomings.
Firstly, no experimental test of the Conway-Kochen Theorem has been performed thus far, and, secondly, because the Conway-Kochen Theorem is strongly connected to the Kochen-Specker Theorem it may be susceptible to the finite precision loophole of Meyer, Kent and Clifton.
In this paper I show that the finite precision loophole does not apply to the Conway-Kochen Theorem.
\end{abstract}

\section{Introduction}
Somewhat loosely speaking, both Bell's Theorem \cite{Bell64,CHSH69} and the (Strong) Free Will Theorem of Conway and Kochen (Conway-Kochen Theorem in this paper) \cite{ConwayKochen06,ConwayKochen09} show the impossibility of local deterministic hidden variable theories for quantum mechanics.
Obviously, if there is a concrete distinction between the two results, it has to lie in the details. 
Recently, Cator \& Landsman \cite{Cator14} made an investigation of these details by reformulating both theorems in a single framework that allows comparison.
They concluded that ``the Strong Free Will Theorem uses fewer assumptions than Bell's 1964 Theorem, as no appeal to probability theory is made, but this comes at a double price. 
Firstly, in the absence of an Aspect-type experiment using spin-one particles, the former so far has no experimental backing. 
Secondly, through its use of the Kochen-Specker Theorem as a lemma, the Strong Free Will Theorem is potentially vulnerable to the kind of `finite precision' challenge discussed (most recently) in \cite{Appleby05,Barrett-Kent04,Hermens11}.''

In this paper I provide a fifty percent discount on the Conway-Kochen Theorem by showing that the finite precision loophole can be closed.
In section \ref{KSsect} the Kochen-Specker Theorem is formulated using the framework introduced in \cite{Cator14}, and I explain that the finite precision loophole attacks a specific assumption that underlies the theorem. 
In section \ref{CKsect} it is highlight that this same assumption also underlies the Conway-Kochen Theorem.
However, I then show that, while relaxing this assumption does allow the possibility of noncontextual deterministic hidden variables (a possibility often alleged to be excluded by the Kochen-Specker Theorem), it does not allow for the possibility of local deterministic hidden variables.
More precisely, I argue that with a slight reformulation of one of the other assumptions underlying the Conway-Kochen Theorem, the theorem still holds.
The main asset of the theorem (that it doesn't make use of probability theory) remains intact.  


\section{Finite Precision and the Kochen-Specker Theorem}\label{KSsect}
\subsection{Introduction}
In a slogan, the Kochen-Specker Theorem \cite{KS67} is often taken to prove ``the impossibility of noncontextual hidden variable theories''.
But as with any foundational result there are footnotes to be placed by the slogan.
One of them is that, actually, the theorem only proves this impossibility under the assumption that a certain set of self-adjoint operators in the quantum formalism corresponds to observables.
And this assumption may be contested.

The idea that not all self-adjoint operators may correspond to observables has been around for some time: \citet{Wigner52} already wondered whether one could meaningfully attribute an observable to operators such as $X+P$ (the sum of the position and momentum operator).
However, it is one thing to argue that not all self-adjoint operators need to correspond to observables, but another to argue of a specific set of self-adjoint operators that they in fact do not correspond to observables.
And it is an argument of the latter type that is needed to disprove the Kochen-Specker slogan.
Such an argument was first provided by \citet{Meyer99} and then further developed in \citep{Barrett-Kent04,CliftonKent99,Hermens11,Kent99}.
It exploits the idea that actual measurements only have a finite precision while the Kochen-Specker Theorem makes use of the idealization of infinite precision.
In this section I quickly rehearse the Kochen-Specker Theorem and the finite precision argument. 
This then allows for a concrete study of how the finite precision argument affects the Conway-Kochen Theorem in section \ref{CKsect}. 

\subsection{Reformulation of the Kochen-Specker Theorem}
The Kochen-Specker Theorem revolves around the quantum theory of spin-1 particles which in turn are modeled on the Hilbert space $\mathbb{C}^3$.
With every unit vector $a$ in $\mathbb{R}^3$ is associated a self-adjoint operator $S_a$ denoting the spin along the $a$-axis.
Since $S_a$ is associated with the axis spanned by $a$ one has $S_a=S_{-a}$.
Thus there is a one-to-one correspondence between spin operators $S_a$ and 1-dimensional projection operators $P_a$ acting on $\mathbb{R}^3$ ($P_av=\langle a,v\rangle a$).
One-dimensional projection operators in turn can be associated with points on the 2-sphere $\mathcal{S}^2$ where opposite points are then being identified with each other. 

A \textit{frame} in $\mathbb{R}^3$ will be an ordered triple $\af=[P_{a_1},P_{a_2},P_{a_3}]$ of 1-dimensional projection operators with $\{a_1,a_2,a_3\}$ an orthonormal basis of $\mathbb{R}^3$.
The set of all frames will be denoted by $\mathcal{F}$.
For every frame $\af\in\mathcal{F}$ the operators $S_{a_1}^2,S_{a_2}^2,S_{a_3}^2$ are pairwise commuting, have spectrum $\{0,1\}$, and sum to two times the identity.
Thus quantum theory predicts that (under the assumption that these operators correspond to observables) a joint measurement of these operators yields a result from the set
\begin{equation}
	T:=\{(1,1,0),(1,0,1),(0,1,1)\}.
\end{equation}

An essential role in Kochen-Specker type arguments is played by frame functions.
A \emph{frame function on} $\mathcal{O}$ is a function $\lambda:\mathcal{O}\subset\mathcal{F}\to T$ such that
\begin{equation}
	\forall \af,\af'\in\mathcal{O}:~\exists i,j ~\text{s.t.}~P_{a_i}=P_{a_j'}~\Rightarrow \lambda_i(\af)=\lambda_j(\af').
\end{equation} 
The idea of a frame function is that it attributes values to operators in a noncontextual way (independent of the frame under consideration).
This is elucidated by noting that every frame function gives rise to a function $c:\mathcal{D}(\mathcal{O})\subset \mathcal{S}^2\to\{0,1\}$ with $\mathcal{D}(\mathcal{O})$ the set of points $a$ for which there is a frame $\af\in\mathcal{O}$ with $P_{a_i}=P_a$ for some $i$ and
\begin{equation}\label{coloring}
	c(a):=\lambda_i(\af).
\end{equation}
This function is an example of a coloring function.
That is, a function $c:\mathcal{D}\subset\mathcal{S}^2\to\{0,1\}$ is called a \emph{coloring function on} $\mathcal{D}$ if 
\begin{enumerate}
\item $c(a)=c(-a)$ whenever $a,-a\in\mathcal{D}$,
\item $c(a_1)+c(a_2)+c(a_3)=2$ whenever $[P_{a_1},P_{a_2},P_{a_3}]$ is a frame and $a_1,a_2,a_3\in\mathcal{D}$,
\item $c(a_1)+c(a_2)\geq1$ whenever $a_1\bot a_2$ and $a_1,a_2\in\mathcal{D}$.
\end{enumerate}

The main mathematical result on which the Kochen-Specker Theorem builds can now be formulated.
\begin{lemma}\label{lemma}
There exists no frame function on $\mathcal{F}$. Equivalently, there is no coloring function on $\mathcal{S}^2$.
\end{lemma} 

Somewhat intuitively, this already establishes the impossibility of attributing definite values to all self-adjoint operators in a noncontextual way.
But the precise philosophical importance can only be highlighted by selecting out concrete assumptions that together require the existence of a frame function on $\mathcal{F}$.
Here I present a set of assumptions that is close the the formulation used in \cite{Cator14}. 
\begin{itemize}
\item \textbf{Determinism} There exists a set $X$ together with a surjective function $A:X\to \mathcal{O}_A\subset\mathcal{F}$ and a function $F:X\to T$. 
\end{itemize}
The associated interpretation is that $\mathcal{O}_A$ selects out the frames $\af$ for which there is an observable corresponding to the joint measurement of $S_{a_1}^2,S_{a_2}^2,S_{a_3}^2$.
With a slight abuse of language I will call the frames in $\mathcal{O}_A$ observables.
An element $x\in X$ then determines both the observable $A(x)$ to be measured and the outcome of the measurement $F(x)$.

\begin{itemize}
\item \textbf{Value Definiteness} There are functions $Z:X\to X_Z$ and $\hat{F}:\mathcal{O}_A\times X_Z\to T$ such that 
\begin{equation}
	F(x)=\hat{F}(A(x),Z(x))~\forall x\in X.
\end{equation}  
\end{itemize}
This assumption establishes that for every $z\in X_Z$ there is a function $\lambda_z:\mathcal{O}_A\to T$ given by 
\begin{equation}
	\lambda_z(\af):=\hat{F}\left(\af,z\right).
\end{equation}
Intuitively $z$ takes on the role of a hidden variable state and the intended reading is that while $x$ only assigns a definite value to the observable that is determined to be measured, the existence of $\hat{F}$ ensures that all observables that are not measured also have a definite value.
But it may also be noted that without any further constraints this assumption is empty as one may just take $X_Z=X$ and $Z=\mathrm{id}$.
Further constraints are then given by the following assumption.\footnote{This way of introducing hidden variable states is perhaps somewhat cumbersome, but it will be helpful for seeing the analogy with the Conway-Kochen Theorem in the next section.} 

\begin{itemize}
\item \textbf{Noncontextuality} For every $z\in X_Z$, $\lambda_z$ is a frame function on $\mathcal{O}_A$.
\end{itemize}
The idea behind this assumption is that the elements of a frame $\af$ or, equivalently, the operators $S_{a_i}^2$ have an ontological status independent of the frame in which they are considered.
That is, the value assigned to $S_{a_i}^2$ via $(\lambda_z(\af))_i$ only depends on $z$ and not on $\af$.

It is this last assumption that is often taken to be ruled out for hidden variable theories by the Kochen-Specker theorem (taking Determinism and Value Definiteness as indispensable ingredients for such a theory).  
In fact, the first argument against the assumption is due to Bell and actually predates the Kochen-Specker Theorem \cite{Bell66}.
In his ``judo-like manoeuvre'' \cite{Shimony84} he argued that Bohr already time and again emphasized that outcomes of experiments cannot be separated from the experimental setup used, and that there is no reason to suppose that this essential ingredient to the Copenhagen interpretation should be something to be denied in a hidden variable theory.
However, this maneuver is not forced upon the proponent of hidden variables by the Kochen-Specker Theorem because, to connect the framework described thus far with Lemma \ref{lemma}, one final assumption is required.

\begin{itemize}
	\item \textbf{Identification Principle} Every frame corresponds to an observable, i.e., $\mathcal{O}_A=\mathcal{F}$.
\end{itemize}
These assumptions together allow the following formulation of the Kochen-Specker Theorem.

\begin{thm}\label{KSthm}
Determinism, Value Definiteness, Noncontextuality and the Identification Principle are mutually exclusive. 
\end{thm}

\subsection{The Finite Precision Loophole}
It should be mentioned that the theorem actually proven by Kochen-Specker is significantly stronger than the one formulated here.
That is, Theorem \ref{KSthm} also holds when the Identification Principle is restricted to a specific finite subset of frames $\mathcal{O}_{KS}$.
That is, the assumption $\mathcal{O}_A=\mathcal{F}$ may be replaced by $\mathcal{O}_A\supset\mathcal{O}_{KS}$.
In the original proof $\mathcal{O}_{KS}$ consists of 133 frames constructed from 117 vectors, and many improvements of the proof have been introduced since that use smaller sets.\footnote{In three dimensions Conway and Kochen hold the record for the smallest number of vectors (31) \cite[p. 114]{Peres02}. The constructions of Peres \cite[p. 198]{Peres02} and Bub \cite{Bub96} both use 33 vectors of which the latter requires the least number of frames of all these results. Recently, the number of contexts has been minimized to 7 for the Hilbert space $\mathbb{C}^6$ \citep{Lisonek14}. For more discussion on comparing sizes of Kochen-Specker sets see \citep{Bengtsson12,PMMM05} and references therein.}
Also, another strengthening of the Kochen-Specker Theorem has been proven recently, which establishes that for any pair of vectors one can construct a finite set of frames for which the theorem holds \cite{Abbott14,Abbott12}.

These results all focus on the minimal requirements for a set $\mathcal{O}\subset\mathcal{F}$ to be such that there does not exist a frame function on $\mathcal{O}$.
What they leave open is the question of how big $\mathcal{O}$ can be such that frame functions on $\mathcal{O}$ \emph{do} exist.
It was shown by Meyer, Kent and Clifton that this set can in fact be quite big:
\begin{thm}\label{MKCthm}
	There exist sets $\mathcal{O}_{MKC}$ that admit frame functions and that are dense in $\mathcal{F}$ in the sense that for every $\epsilon>0$, for every $\af\in\mathcal{F}$ there is an $\af'\in\mathcal{O}_{MKC}$ such that
\begin{equation}
	\max_i\min_j\|P_{a_i}-P_{a'_j}\|<\epsilon.
\end{equation}
Consequently, $\mathcal{D}(\mathcal{O}_{MKC})$ is dense in $\mathcal{S}^2$ in the usual sense.
\end{thm} 
Meyer showed in \cite{Meyer99} that one can take the set of frames 
\begin{equation}
	\left\{[P_{a_1},P_{a_2},P_{a_3}]\in\mathcal{F}\:\middle|\:a_i\in\mathbb{Q}^3\right\}.
\end{equation}
This result was further generalized by Kent and Clifton who showed that for any finite-dimensional Hilbert space a set $\mathcal{O}_{MKC}$ can be constructed such that all frames in it are totally incompatible \cite{CliftonKent99,Kent99}.
This means that for every unit vector $a_i$ there is at most one frame in $\mathcal{O}_{MKC}$ containing $P_{a_i}$.
The existence of frame functions on $\mathcal{O}_{MKC}$ then becomes a triviality. 
Furthermore, it allows one to show that there are enough frame functions such that every quantum state can be represented by a probability distribution over these frame functions \cite{CliftonKent99,Hermens11}. 

At first sight it looks like one arrives at a stalemate concerning the possibility of noncontextual hidden variable theories.
If $\mathcal{O}_{KS}\subset\mathcal{O}_A$ for some Kochen-Specker set $\mathcal{O}_{KS}$, then noncontextuality fails.
If on the other hand $\mathcal{O}_A\subset\mathcal{O}_{MKC}$, then noncontextuality is a possibility.
The finite precision argument is the tiebreaker here.
The idea is that due to the finite precision of measurements, one cannot establish with infinite precision which frame is selected out from $\mathcal{F}$ when a measurement is performed.
Thus on this view the Identification Principle can be weakened to obtain something like the following.
\begin{itemize}
\item \textbf{Identification Principle}$_{\textbf{FP}}$ For every frame $\af\in\mathcal{F}$ there is an observable $\af'\in\mathcal{O}_A$ such that $\af$ and $\af'$ are empirically indistinguishable.
\end{itemize}
Theorem \ref{MKCthm} establishes that this principle can be maintained: no matter how precisely a frame is determined experimentally, there are always frames within $\mathcal{O}_{MKC}$ compatible with it.
The set of frames $\mathcal{O}_{MKC}$ is thus rich enough to reproduce all the empirical predictions one could make under the assumption that $\mathcal{F}$ is the set of `real' frames.
The upshot is that the finite precision loophole allows the construction of noncontextual hidden variable theories that respect the Kochen-Specker Theorem by giving up the Identification Principle.
In the next section I show that even though the same principle is adopted in the proof of the Conway-Kochen Theorem, it can not be exploited to undermine the import of the theorem, being, the impossibility of a deterministic local theory that is compatible with `free will'.
But before discussing this result I want to make two final remarks.

First, the noncontextual MKC-models are surely counterintuitive and therefore have yielded quite some criticism in a wide variety.
Some have tried to dismiss the MKC-models as viable possibilities by arguing that they cannot reproduce all empirical predictions of quantum mechanics (such as \cite{Cabello02}), while others have focused more on the unsatisfactory aspects of the models themselves (e.g. \cite{Appleby05}).
For the details on these objections and possible responses I refer the reader to the discussions in \citep{Appleby05,Barrett-Kent04,Held13,Hermens11} and references therein. 
However, there is one recent objection that deserves a quick response 
It focuses on the observation that the restriction to a proper subset of frames may be interpreted as the impossibility to align a spin-measurement device along a certain axis in real space.
In \citep{Held13} this is taken to indicate the \emph{non-existence} of these directions in real space.
While this is a subtle step that may deserve some debate, I'd only like to point out that the continuity of space itself is not an innocent assumption to be upheld for a hidden variable theory.
Furthermore, with the advancements in current physics, it is not unlikely that this assumption is to be dropped anyway as for example in loop quantum gravity (c.f. \citep{Smolin00} for a friendly conceptual introduction.).   

Second, the Kochen-Specker Theorem may also be avoided by dropping one of the assumptions Determinism or Value Definiteness. 
And one can argue for the possibility of noncontextual interpretations of quantum mechanics along those lines.
Although on the present presentation this would require a reformulation of Noncontextuality (because it relies on the earlier assumptions being true), a sketch of such arguments can still be given.
Consider someone who takes Determinism seriously. 
Then the state of the world determines which observable is to be measurement and the outcome of this measurement.
And nothing more seems to be needed to explain what we actually observe.
There is then no reason to assume that unmeasured observables also should have a definite value (i.e., Value Definiteness is abandoned).
In fact, one may argue that unmeasured observables are not observables at all, and the value attribution given by Determinism may then be considered to be noncontextual.
Another possibility is to argue that Determinism fails as it does for example in the consistent histories approach.
For an argument that Noncontextuality can be saved on this view see \citep{Griffiths13}. 
This all demonstrates that the slogan that ``the Kochen-Specker Theorem shows that Nature is contextual'' is even more reckless than the one I started with in this section. 

\section{Finite Precision and the Conway-Kochen Theorem}\label{CKsect} 
\subsection{The Conway-Kochen Theorem}
The Conway-Kochen Theorem \cite{ConwayKochen06,ConwayKochen09} is preceded  by a long history of investigations making use of the main ingredient of the theorem: a system of two entangled spin-1 particles.
This history starts with \citet{HeywoodRedhead83} who used this system to investigate the relation between Noncontextuality (as assumed in the Kochen-Specker Theorem) and Locality (as assumed in the derivation of any Bell-inequality).
Since then it has also occurred in \cite{BrownSvetlichny90,Clifton93,Stairs83} and (unsurprisingly) it has been claimed that Conway and Kochen haven't proven anything that wasn't already known \citep{Goldstein10}.
I tend to disagree with this view (c.f. \citep[\S 4.4.4]{Hermens10}) and this disagreement is backed up by \citet{Cator14}.
What their paper shows is that the Conway-Kochen Theorem can be understood as an adaption of the Kochen-Specker Theorem to provide an argument against local determinism that is \emph{independent} of the Bell-type argument instead of a reformulation or investigation thereof.  
This independence comes with a prize, and one of them is that the theorem may be susceptible to the finite precision loophole.

To investigate this, the theorem has to be formulated in a way that elucidates where the Identification Principle comes in as an assumption.
This is done by making a slight adaption to the assumptions introduced in \cite{Cator14}.
The setting is a two agent version of the Kochen-Specker Theorem discussed in the previous section.
\begin{itemize}
\item \textbf{Determinism} There exists a set $X$ together with surjective functions 
\begin{equation}	
	A:X\to \mathcal{O}_A\subset\mathcal{F},~B:X\to \mathcal{O}_B\subset\mathcal{F}
\end{equation}
and functions $F,G:X\to T$ where $A$ and $B$ represent the performed measurements and $F$ and $G$ their outcomes.
\item \textbf{Parameter Independence} There are functions $Z:X\to X_Z$ and 
\begin{equation}
	\hat{F}:\mathcal{O}_A\times X_Z\to T,~\hat{G}:\mathcal{O}_B\times X_Z\to T
\end{equation} 
such that 
\begin{equation}
	F(x)=\hat{F}(A(x),Z(x)),~G(x)=\hat{G}(B(x),Z(x))~\forall x\in X.
\end{equation} 
\item \textbf{Freedom} $(A,B,Z)$ are independent in the sense that for all $(\af,\bv,z)\in\mathcal{O}_A\times\mathcal{O}_B\times X_Z$ there is an $x\in X$ such that $(A(x),B(x),Z(x))=(\af,\bv,z)$, i.e., the three component function $(A,B,Z)$ is surjective. 
\item \textbf{Identification Principle} Every frame corresponds to an observable: $\mathcal{O}_A=\mathcal{O}_B=\mathcal{F}$.
\end{itemize}

When comparing with the Kochen-Specker Theorem, one finds that Parameter Independence replaces Value Definiteness.
And though, like Value Definiteness, Parameter Independence is an empty assumption without any further constraints, it is in fact a stronger assumption.
Value Definiteness for the two particle system only requires the existence of a function $\hat{V}:\mathcal{O}_A\times\mathcal{O}_B\times X_Z\to T\times T$ such that $\hat{V}(A(x),B(x),Z(x))=(F(x),G(x))$ while Parameter Independence requires in addition that $\hat{V}$ factorizes in two functions $\hat{F}$ and $\hat{G}$.
It thus embodies that the settings $\af$ and $\bv$ can be selected independently.
The freedom assumption further establishes that this selection is free in the sense that it is not constrained by the definite values assigned to each setting by $x$.    

Finally, the Conway-Kochen Theorem requires an empirical assumption: 
\begin{itemize}
\item \textbf{Nature} For any pair of frames $\af=[P_{a_1},P_{a_2},P_{a_3}],\bv=[P_{b_1},P_{b_2},P_{b_3}]$ if there are $i,j$ such that $P_{a_i}=P_{b_j}$, then
\begin{equation}
	\hat{F}_i(\af,z)=\hat{G}_j(\bv,z)~\forall z\in X_Z.
\end{equation}
\end{itemize}
The Conway-Kochen Theorem then states:
\begin{thm}\label{thmckfp}
Determinism, Parameter Independence, Freedom, the Identification Principle and Nature are mutually exclusive.
\end{thm}

\subsection{Closing the Finite Precision Loophole}
Although the proof for Theorem \ref{thmckfp} may be found in \cite{Cator14}, to get a better grip on the role of the Identification Principle it is useful to prove it again with a small detour.

\begin{lemma}\label{lemma2}
If $\mathcal{O}_A=\mathcal{O}_B$, then Determinism, Parameter Independence, Freedom and Nature imply the existence of a frame function on $\mathcal{O}_A$.
\end{lemma}
\begin{proof}
First of all, Determinism and Parameter Independence together imply that for every $z\in X_Z$
\begin{equation}
	\lambda_z:\mathcal{O}_A\to T,~\lambda_z(\af):=\hat{F}(\af,z)
\end{equation}
is well-defined.
It remains to be proven that $\lambda_z$ is a frame function on $\mathcal{O}_A$.
Now suppose that $\af,\af'\in\mathcal{O}_A$ and there are $i,j$ such that $P_{a_i}=P_{a'_j}$.
Since $\mathcal{O}_A=\mathcal{O}_B$ there are $\bv,\bv'\in\mathcal{O}_B$ such that $\af=\bv,\af'=\bv'$.
Because of Freedom for every $z\in X_Z$ there are $x,x'\in X$ such that $(A(x),B(x),Z(x))=(\af,\bv,z)$ and $(A(x'),B(x'),Z(x'))=(\af',\bv,z)$.
It then follows from Nature that
\begin{equation}
	\lambda_z(\af)_i=\hat{F}_i(\af,z)=\hat{G}_i(\bv,z)=\hat{F}_j(\af',z)=\lambda_z(\af')_j.
\end{equation} 
Thus $\lambda_z$ is a frame function on $\mathcal{O}_A$.
\end{proof}

The proof of Theorem \ref{thmckfp} now follows from this lemma together with Lemma \ref{lemma} and the Identification Principle.
It is thus tempting to believe that replacing the Identification Principle with its finite precision version may be sufficient to block the proof.
Formally, this is indeed true.
If one takes $\mathcal{O}_A=\mathcal{O}_B=\mathcal{O}_{MKC}$ all premises can be satisfied.
The Nature premise then amounts to the restriction that $x$ is such that 
\begin{equation}
	\hat{F}(\af,z)=\hat{G}(\bv,z)
\end{equation}
whenever $\af=\bv$.
This sparks the idea that there is an even more trivial way to ensure consistency: one could choose $\mathcal{O}_A$ and $\mathcal{O}_B$ such that they have no one-dimensional projections in common.
That is, they are chosen such that the antecedent of Nature is never fulfilled. 

This possibility demonstrates that the formulation of Nature is more strict than its intended meaning.
Nature is supposed to be an empirical constraint and this requires that the antecedent poses an experimental possibility.
Here one finds that the finite precision argument is a double-edged sword: not only does it allow for a weakening of the Identification Principle, but it also justifies a strengthening of the Nature assumption.
Somewhat loosely, it allows the following reformulation.
\begin{itemize}
\item \textbf{Nature}$_\textbf{FP}$ For any pair of frames $\af=[P_{a_1},P_{a_2},P_{a_3}]$, $\bv=[P_{b_1},P_{b_2},P_{b_3}]$, if there are $i,j$ such that $\|P_{a_i}-P_{b_j}\|$ is small, then in most of the cases
\begin{equation}
	\hat{F}_i(\af,z)=\hat{G}_j(\bv,z).
\end{equation} 
\end{itemize}
In fact, an experimental validation of Nature would imply an experimental validation of Nature$_{\text{FP}}$.

Roughly, this amounts to the demand that the coloring function $c_z$ generated by $\lambda_z$ via \eqref{coloring} is continuous on most of the points in $\mathcal{D}(\mathcal{O}_A)$.
Now it doesn't matter much in which way one wishes to make the notion of `most of the points' precise; any reasonable definition will lead to trouble.
This is because it was shown by \citet{Appleby05} that for every coloring function $c$ on a dense subset $\mathcal{D}\subset\mathcal{S}^2$ there is a non-empty open region $D\subset\mathcal{S}^2$ such that $c$ is densely discontinuous on $\mathcal{D}\cap D$.
That is, for every point $a\in \mathcal{D}\cap D$ and every neighborhood $U$ of $a$ there is a point $a'\in U\cap \mathcal{D}\cap D$ such that $c(a)\neq c(a')$.
Then, whenever Determinism, Parameter Independence and Freedom are taken to hold, and $\mathcal{O}_{A}$ and $\mathcal{O}_B$ are assumed to be dense in $\mathcal{F}$ in the MKC-sense, Nature$_{\text{FP}}$ fails. 
This is because no matter how precise the directions of a triad can be determined, there are always frames $\af,\af'\in\mathcal{O}_A$ that are so close to each other that they cannot be distinguished empirically and such that 
\begin{equation}
	\hat{F}_i(\af,z)\neq\hat{F}_j(\af',z),
\end{equation} 
where $a_i$ and $a'_j$ are the directions that align with each other. In short, one has the following:
\begin{thm}\label{thmckfp2}
Determinism, Parameter Independence, Freedom, the Identification Principle$_{\text{FP}}$ and Nature$_{\text{FP}}$ are mutually exclusive.
\end{thm}

\section{Remarks and Conclusion}
It has been shown that although formally the finite precision loophole that applies to the Kochen-Specker Theorem also applies to the Conway-Kochen Theorem, this loophole can be closed once it is recognized that the Nature assumption is an empirical assumption rather than one of principle (which does apply to the other assumptions).
Consequently, if an experimental test of the Conway-Kochen Theorem were to be performed (which would constitute an empirical validation of Nature for some set of frames (approximately) within an uncolorable set), the test would also close the finite precision loophole.
This reminds one of the second disadvantage of the Conway-Kochen Theorem as noted by Cator \& Landsman: thus far no experimental test of the theorem has been performed.
If, however, one assumes confidence in the possibility of experimentally testing the theorem and that the results will be in favor of it, then I can only agree with their conclusion that the theorem has a big advantage over Bell's Theorem by not using probability theory.
Indeed, the role of probability in Bell's Theorem has played a significant (albeit sometimes implicit) role in the discussion of what the theorem actually tells us about the world.\footnote{A very incomplete list of examples is \cite{Atkinson01,Fine82,Fine82a,Jones94,Khrennikov09,Brown88,Santos91}.}
Hence it is safe to conclude that the Conway-Kochen Theorem does provide a valuable contribution to the foundational debate even though this wasn't immediately clear when it was first published. 

\section*{Acknowledgments}
I would like to thank N. P. Landsman for encouraging comments on an early sketch of this paper.
This work was supported by the NWO (Vidi project nr. 016.114.354).

\printnomenclature

\printbibliography

\end{document}